\documentclass{proceedings}

\usepackage[hidelinks]{hyperref}  %% makes hyper references active
\usepackage[anythingbreaks]{breakurl}  %% allows line breaks everywhere
\usepackage[all]{nowidow} %% prevents single line after page brake
\usepackage{amsthm}

\usepackage[caption=false, position=top, singlelinecheck=false]{subfig} %% supports sub-figures and sub-tables
\newtheorem{theorem}{Theorem}
\newtheorem{proposition}{Proposition}
\newtheorem{lemma}{Lemma}

\newtheorem{defn}{Definition}[section]
\usepackage{xcolor}
\usepackage[normalem]{ulem}
\usepackage{color,soul}
\usepackage{stix}
\def\mathunderline#1#2{\color{#1}\underline{{\color{black}#2}}\color{black}}

\def\mathdashuline#1#2{\color{#1}\dashuline{{\color{black}#2}}\color{black}}

\def\mathdotuline#1#2{\color{#1}\dotuline{{\color{black}#2}}\color{black}}
\title{Odd and even derivations, transposed Poisson superalgebra and 3-Lie superalgebra}

\author{Viktor Abramov and Nikolai Sovetnikov\corauthref{Corresponding author, nikolai.sovetnikov@ut.ee}}

\address{Institute of Mathematics and Statistics, University of Tartu, Narva mnt 18, 51009 Tartu, Estonia}

\abstract{One important example of a transposed Poisson algebra can be constructed by means of a commutative algebra and its derivation. This approach can be extended to superalgebras, that is, one can construct a transposed Poisson superalgebra given a commutative superalgebra and its even derivation. In this paper we show that including odd derivations in the framework of this approach requires introducing a new notion. It is a super vector space with two operations that satisfy the compatibility condition of transposed Poisson superalgebra. The first operation is determined by a left supermodule over commutative superalgebra and the second is a Jordan bracket. Then it is proved that the super vector space generated by an odd derivation of a commutative superalgebra satisfies all the requirements of introduced notion. We also show how to construct a 3-Lie superalgebra if we are given a transposed Poisson superalgebra and its even derivation.}

\keywords{Commutative superalgebra; Lie superalgebra; transposed Poisson superalgebra; 3-Lie superalgebra; Jordan algebra and Jordan module}

\markleft{Viktor Abramov and Nikolai Sovetnikov.: 3-Lie superalgebra}  %% running title if two authors
\begin{document}

\maketitle

%\begin{multicols}{2}  %% uncomment for two-column layout
%%%%%%%%%%%%%%%%%%%%%%%%%%%%%%%%
%%%%%%%%%%%%%%%%%%%%%%%%%%%%%%%%%
\section{INTRODUCTION}
%%%%%%%%%%%%%%%%%%%%%%%%%%%%%%%%%%%%
%%%%%%%%%%%%%%%%%%%%%%%%%%%%%%%%%%%%
%%%%%%%%%%%%%%%%%%%%%%%%%%%%%%%%%%%%
One of the most important structures of the mathematical methods of Hamiltonian mechanics is a Poisson bracket. If the coordinates of $2n$-dimensional canonical phase space are denoted by $(p_i,q_i)$, where integer $i$ runs from 1 to $n$, then the Poisson bracket of two smooth functions $f,g$ can be written as follows
\begin{equation}
\{f,g\}=\sum_{i=1}^n\;\Big(\frac{\partial f}{\partial q_i}\frac{\partial g}{\partial p_i}-\frac{\partial f}{\partial p_i}\frac{\partial g}{\partial q_i}\Big).
\label{canonical Poisson bracket}
\end{equation}
From (\ref{canonical Poisson bracket}) it follows that the Poisson bracket is skew-symmetric, that is, $\{f,g\}=-\{g,f\}$ and it satisfies the Jacobi identity 
$$
\big\{\{f,g\},h\big\}+\big\{\{g,h\},f\big\}+\big\{\{h,f\},g\big\}=0.
$$
Thus the vector space of smooth functions endowed with the Poisson bracket (\ref{canonical Poisson bracket}) is a Lie algebra \cite{Arnold,Hall}. But there is another structure on the vector space of smooth functions. If we consider pointwise product of two smooth functions $fg$, then the vector space of smooth functions becomes an associative commutative algebra and we have 
\begin{equation}
\{fg,h\}=f\{g,h\}+\{f,h\}g.
\label{differentiation by Poisson bracket}
\end{equation}
The Poisson bracket of two smooth functions defined on a phase space leads to an algebraic structure called a Poisson algebra. A Poisson algebra $(P,\cdot,\{\;,\;\})$ is a vector space $P$ with two binary operations, where $(P,\cdot)$ is a commutative associative algebra and $(P,\{\;,\;\})$ is a Lie algebra. Additionally these two binary operations satisfy the relation
\begin{equation}
\{z,x\cdot y\}=\{z,x\}\cdot y+x\cdot\{z,y\},\;\;x,y,z\in P,
\label{intr-comp-Poisson}
\end{equation}
which is usually called a compatibility condition.

A notion of transposed Poisson algebra was introduced and studied in \cite{Bai-Guo-Wu}. By definition a transposed Poisson algebra $({\cal P},\cdot, \{\;,\;\})$, like a Poisson algebra, is a vector space ${\cal P}$ with two binary operations, where the first is a commutative associative multiplication $(x,y)\in{\cal P}\times{\cal P}\to x\cdot y\in {\cal P}$ and the second is a Lie bracket $(x,y)\in{\cal P}\times{\cal P}\to \{x,y\}\in {\cal P}$. But a compatibility condition in the case of a transposed Poisson algebra has a form different from the compatibility condition (\ref{intr-comp-Poisson}) for a Poisson algebra. The compatibility condition in the case of a transposed Poisson algebra has the form
\begin{equation}
2\,z\cdot\{x,y\}=\{z\cdot x,y\}+\{x,z\cdot y\}.
\label{intr-trans-comp-condition}
\end{equation}
The compatibility condition of a transposed Poisson algebra shows that for any element $x$ of a transposed Poisson algebra the operator of multiplication from the left $L_x(y)=x\cdot y$ is similar to derivation of a Lie bracket (if we do not take into account the factor 2 on the left). In this sense condition (\ref{intr-trans-comp-condition}) is the transpose of condition (\ref{intr-comp-Poisson}), and this makes it appropriate to use the term "transposed" in the name of an algebra. An example of a transposed Poisson algebra can be constructed by means of a derivation of a commutative associative algebra \cite{Bai-Guo-Wu}. If $(A,\cdot)$ is a commutative associative algebra and $D:A\to A$ is a derivation then the bracket
\begin{equation}
[x,y]=x\cdot D(y)-y\cdot D(x),\;\;x,y\in A,\label{intr-D-bracket}
\end{equation}
is a Lie bracket and it satisfies the compatibility condition (\ref{intr-trans-comp-condition}). Hence $(A,\cdot,[\;,\;])$ is a transposed Poisson algebra. An overview of recent results in transposed Poisson algebras can be found in \cite{Kaygorodov}.

The notion of a transposed Poisson algebra can be extended to superalgebras if in the definition of a transposed Poisson algebra we assume that $({\cal P},\cdot)$ is a commutative associative superalgebra, and $({\cal P},\{\;,\;\})$ is a Lie superalgebra \cite{Musson}. Using the Koszul sign rule, we write the compatibility condition (\ref{intr-trans-comp-condition}) in the graded form
\begin{equation}
2\,z\cdot\{x,y\}=\{z\cdot x,y\}+(-1)^{|x||z|}\{x,z\cdot y\},
\label{intr-graded-comp-condition}
\end{equation}
where $|x|,|z|$ are the parities of elements of a superalgebra. The notion of a transposed Poisson superalgebra was introduced in \cite{Abramov 1} and later the same definition was given in \cite{Ouaridi}, where the author proved that Kantor double of a transposed Poisson algebra is a Jordan superalgebra. In \cite{Abramov 1} it was shown that one can construct an important example of a transposed Poisson superalgebra using an even derivation of a commutative superalgebra in a similar way to how this is done in the case of a transposed Poisson algebra (\ref{intr-D-bracket}). In more detail, if we have a commutative superalgebra $(A,\cdot)$ and its even derivation $D$, then we define the bracket by the formula 
\begin{equation}
[x,y]=x\cdot D(y)-(-1)^{|x||y|}\,y\cdot D(x),
\label{intr-graded-D-bracket}
\end{equation}
where $x,y$ are elements of a commutative superalgebra, $|x|,|y|$ are their parities and $D:A\to A$ is an even derivation. Then $(A,[\;,\;])$ is a Lie superalgebra called a virasorisation of $A$ \cite{Roger}. Now it can be proved that $(A,\cdot,[\;,\;])$ is a transposed Poisson superalgebra \cite{Abramov 1}. It is worth to note that it is not possible to use bracket (\ref{intr-graded-D-bracket}) in the case when $D$ is an odd derivation, because in this case the consistency of parities is broken, that is, the parity of the right-hand side of the equation (\ref{intr-graded-D-bracket}) is $|x|+|y|+1$, but it should be $|x|+|y|$. In this paper we investigate the question of how an odd derivation of a commutative superalgebra can be used to construct (analogously to (\ref{intr-graded-D-bracket})) a structure similar to a transposed Poisson superalgebra. To this end, we consider a left $A$-supermodule ${\cal E}={\cal E}_0\oplus{\cal E}_1$, where $A=A_0\oplus A_1$ is a commutative superalgebra. We assume that a left $A$-supermodule $\cal E$ is endowed with a bilinear bracket $(X,Y)\in{\cal E}\times{\cal E}\to\{X,Y\}\in{\cal E}$ which is commutative $\{X,Y\}=(-1)^{|X||Y|}\{Y,X\}$ ($|X|,|Y|$ are the parities) and satisfies $|\{X,Y\}|=|X|+|Y|$. Next we assume that $({\cal E}_0,\{\;,\;\})$ is a Jordan algebra \cite{McCrimmon} and the mappings ${\cal E}_0\times{\cal E}_1\to {\cal E}_1,{\cal E}_1\times{\cal E}_0\to {\cal E}_1$ defined by $(X,Y)\to \{X,Y\}$ and $\{Y,X\}\to \{Y,X\}$, where $X\in{\cal E}_0,Y\in{\cal E}_1$, define on ${\cal E}_1$ the structure of a Jordan module over ${\cal E}_0$ \cite{Carotenuto} . Then we propose to call a left $A$-supermodule structure of $\cal E$ transposed Poisson type compatible (TP-compatible in short) with a Jordan structure of $\cal E$ if they satisfy the super transposed Poisson compatibility condition (\ref{intr-graded-comp-condition}) written in this case in the form
\begin{equation}
2\,z\cdot\{X,Y\}=\{z\cdot X,Y\}+(-1)^{|z||X|}\{X,z\cdot Y\},\;\;z\in A,\;\; X,Y\in\cal E,
\label{intr-graded-comp-condition 2}
\end{equation}
where dot stands for multiplication of elements of left $A$-supermodule $\cal E$ by elements of a superalgebra $A$.

As an example of TP-compatible structure, we consider a commutative superalgebra $A$, its odd derivation $\delta$ and the super vector space ${\mathfrak D}=\{x\cdot \delta:x\in A\}$, where $(x\cdot\delta)(y)=x\,\delta(y)$ and $|x\cdot\delta|=|x|+1$. Obviously ${\mathfrak D}={\mathfrak D}_0\oplus{\mathfrak D}_1$ is a left $A$-supermodule, where ${\mathfrak D}_0\cong A_1, {\mathfrak D}_1\cong A_0$. We endow $\mathfrak D$ with the bracket \cite{Roger}
\begin{equation}
\{X,Y\}=(x\,\delta(y)+(-1)^{(|x|+1)(|y|+1)}y\,\delta(x))\cdot\delta,
\label{intr-graded-delta-bracket}
\end{equation}
where $X=x\cdot\delta, Y=y\cdot\delta$. We prove that $({\mathfrak D}_0,\{\;,\;\})$ is a Jordan algebra and ${\mathfrak D}_1$ is a Jordan module over ${\mathfrak D}_1$. Then we prove that the left $A$-supermodule structure of $\mathfrak D$ is TP-compatible with the Jordan structure of $\mathfrak D$ induced by the bracket (\ref{intr-graded-delta-bracket}), that is, we prove the TP-compatibility condition (\ref{intr-graded-comp-condition 2}).

The second question we consider in this paper is the the question of how one can construct a 3-Lie superalgebra given a transposed Poisson superalgebra and its even derivation. In paper \cite{Bai-Guo-Wu} the authors showed a very important connection between transposed Poisson algebras and 3-Lie algebras. In particular, it was shown that given a transposed Poisson algebra and its derivation, one can construct a ternary bracket that satisfies all the requirements of a 3-Lie algebra. The concept of $n$-Lie algebra was introduced by Filippov in \cite{Filippov}. Later, a surge of interest in these algebras \cite{Abramov 1, Abramov 2, RBai, Rotkiewicz} and their various generalizations was due to their applications in generalized Hamiltonian mechanics \cite{Nambu,Awata} and M-brain theory \cite{Lambert,Palmer} . An $n$-Lie superalgebra is a $\mathbb Z_2$-extension of the concept of $n$-Lie algebra to the super case \cite{Abramov 2, Abramov 3}. Our approach to the question mentioned above is based on the construction proposed in \cite{Bai-Guo-Wu}, where it was shown how one can construct a 3-Lie algebra given a transposed Poisson algebra and its derivation. To be more exact, we prove that if $\cal P$ is a transposed Poisson superalgebra and $D$ is an even derivation of this algebra then the ternary bracket
$$
[x,y,z]=D(x)\cdot [y,z]+(-1)^{|x|(|y|+|z|)}D(y)\cdot[z,x]+(-1)^{(|x|+|y|)|z|}D(z)\cdot[x,y],
$$
satisfies the super Filippov-Jacobi identity or, by other words, $({\cal P},[\;,\;,\;])$ is a 3-Lie superalgebra. This theorem was proved in \cite{Sovetnikov} and we present this proof in Section 4 with some minor modifications.
%%%%%%%%%%%%%%%%%%%%%%%%%%%%%%%%%%%
%%%%%%%%%%%%%%%%%%%%%%%%%%%%%%%%%%%
\section{Preliminaries}
%%%%%%%%%%%%%%%%%%%%%%%%%%%%%%%%%%%%%
%%%%%%%%%%%%%%%%%%%%%%%%%%%%%%%%%%%
%%%%%%%%%%%%%%%%%%%%%%%%%%%%%%%%%%%%
In this paper we investigate the structure of a transposed Poisson superalgebra and how it can be used to construct 3-Lie superalgebras. Throughout what follows, field $\mathbb K$ means either the field of real numbers $\mathbb R$ or the field of complex numbers $\mathbb C$. A Poisson algebra $(P,\cdot,[\;,\;])$ is a $\mathbb K$-vector space $P$ with two binary operations, where $( P,\cdot)$ is an associative commutative algebra and $(P,[\;,\;])$ is a Lie algebra. In addition, these two structures must be compatible, that is, the following condition must be satisfied
\begin{equation}\label{comp-Poisson}
    [x,y\cdot z] = [x,y]\cdot z + y\cdot [x,z] \qquad \forall x,y,z\in P.
\end{equation}
Thus, the compatibility condition (\ref{comp-Poisson}) shows that each element of $P$ defines, by means of a Lie bracket, a derivation of an associative commutative algebra $P$. The well known example of a Poisson algebra is the associative, commutative algebra of smooth functions on a phase space, equipped with a Poisson bracket. A derivation of a Poisson algebra is a linear mapping $D:P\to P$ which is a derivation of an associative commutative algebra $(P,\cdot)$ as well as a derivation of a Lie algebra $(P,[\;,\;])$. Thus, in the case of derivation of a Poisson algebra we have
\begin{equation}\label{Leibniz-for-dot}
    D(x\cdot y) = D(x)\cdot y + x\cdot D(y),\;\;\; D([x,y]) = [D(x),y]+[x,D(y)].
\end{equation}
The notion of transposed Poisson algebra was first proposed and studied in \cite{Bai-Guo-Wu}. This structure turned out to be very interesting from a point of view of algebra and geometry as well as form a point of view of applications in mathematical physics. A transposed Poisson algebra $({\cal P},\cdot,[\,,\,])$, like a Poisson algebra, is a $\mathbb K$-vector space $\cal P$ with two binary operations, where $(\cal P,\cdot)$ is an associative commutative algebra and $(\cal P,[\;,\;])$ is a Lie algebra. The difference between the notion of a Poisson algebra and that of a transposed Poisson algebra is a compatibility condition. In the case of a transposed Poisson algebra $\cal P$, the compatibility condition has the form
\begin{equation}\label{comp-transposed}
    2x\cdot [y,z] = [x\cdot y,z]+[y,x\cdot z] \qquad \forall x,y,z\in \cal P.
\end{equation}
This condition shows that each element of $\cal P$ determines, by means of a multiplication of an associative commutative algebra $({\cal P},\cdot)$, a derivation of a Lie bracket. Hence the use of the term "transposed" in this context is relevant. We also note one more difference between the compatibility condition (\ref{comp-transposed}) and the compatibility condition (\ref{comp-Poisson}). This is the presence of the factor 2 on the left-hand side of (\ref{comp-transposed}). This factor plays an important role and is related to an arity of a Lie bracket. An important example of a transposed Poisson algebra proposed in \cite{Bai-Guo-Wu} is an associative commutative algebra $(A,\cdot)$ equipped with the Lie bracket
\begin{equation}
[x,y]=x\,D(y)-y\,D(x), \;\;\;\;x,y\in A,
\label{bracket with derivation}
\end{equation}
where $D$ is a derivation of $A$.

We can extend the notion of a transposed Poisson algebra to superalgebras assuming that $({\cal P},\cdot)$ is a commutative associative superalgebra and $({\cal P},[\;,\;])$ is a Lie superalgebra. As a compatibility condition for these two structures, it is natural to take a graded version of the compatibility condition (\ref{comp-transposed}) for a transposed Poisson algebra. The notion of a transposed Poisson superalgebra was introduced in \cite{Abramov 1}. Let ${\cal P}={\cal P}_0\oplus {\mathcal P}_1$ be a super $\mathbb K$-vector space. The parity of a homogeneous element $x\in{\cal P}$ will be denoted by $|x|$, that is, $|x|\in{\mathbb Z}_2$. 
\begin{defn}
     A transposed Poisson superalgebra is a triple $(\cal P, \cdot, [\;,\;])$, where $(\cal P,\cdot)$ is a commutative associative superalgebra, $(\cal P,[\,,\,])$ is a Lie superalgerba and for any three elements $x,y,z\in {\cal P}$ it holds 
    \begin{equation}\label{comp-super-transposed}
        2z\cdot[x,y] = [z\cdot x, y] + (-1)^{|x||z|}[x,z\cdot y]. 
    \end{equation}
\end{defn}
\noindent
In \cite{Abramov 1} it was shown that a transposed Poisson superalgebra can be constructed in a similar way to a transposed Poisson algebra, which is constructed on an associative commutative algebra by means of a derivation and the Lie bracket (\ref{bracket with derivation}). But in the case of a superalgebra one can attribute a parity to a derivation. We remind that a derivation $D$ of a superalgebra $A$ is called even if it does not change the parity of a homogeneous element, that is,  $|D(x)|=|x|$ for any $x\in A_0\cup A_1$, and a derivation $D$ of a superalgebra $A$ is called odd if it changes the parity of a homogeneous element, that is, $|D(x)| = |x|+1$ for arbitrary $x\in A_0\cup A_1$. The parity of a derivation $D$ will be denoted by $|D|$. Hence, $|D|=0$ if $D$ is an even derivation and $|D|=1$ if $D$ is odd. The case of even derivation was considered in \cite{Abramov 1} and it was shown that if $(A,\cdot)$ is a commutative associative superalgebra, $D:A\to A$ is an even derivation and the bracket is defined by
\begin{equation}
[x,y]=x\cdot D(y)-(-1)^{|x||y|}\;y\cdot D(x),\;\;x,y\in A,
\label{graded bracket with even derivation}
\end{equation}
then $(A,\cdot,[\;,\;])$ is a transposed Poisson superalgebra. The case of an odd derivation will be considered in the present paper in the next section.

One of the goals of this paper is to construct a 3-Lie superalgebra given a transposed Poisson superalgebra and its even derivation. For this we need a graded form of the identities proved in \cite{Bai-Guo-Wu}. The graded form of these identities is given in the following theorem proved in \cite{Abramov 1}. 
\begin{theorem}
Let $(\cal P,\cdot,[\;,\;])$ be a transposed Poisson superalgebra. Then for any $h,x,y,z,u,v\in \cal P$ we have the following identities
\begin{eqnarray}
(-1)^{|x|\,|z|}\,x\cdot[y,z]+(-1)^{|x|\,|y|}\,y\cdot[z,x]+(-1)^{|y|\,|z|}\,z\cdot[x,y]\!\!\! &=&\!\!\! 0,\label{identity 1}\\
(-1)^{|x|\,|z|}\,[h\cdot[x,y],z]+(-1)^{|x|\,|y|}\,[h\cdot[y,z],x]+(-1)^{|y|\,|z|}\,[h\cdot[z,x],y]\!\!\! &=&\!\!\! 0,\label{identity 2}\\
(-1)^{|x|\,|z|}\,[h\cdot x,[y,z]]+(-1)^{|x|\,|y|}\,[h\cdot y,[z,x]]+(-1)^{|y|\,|z|}\,[h\cdot z,[x,y]]\!\!\! &=&\!\!\! 0,\label{identity 3}\\
(-1)^{|x|\,|z|}\,[h,x]\,[y,z]+(-1)^{|x|\,|y|}\,[h,y]\,[z,x]+(-1)^{|y|\,|z|}\,[h,z]\,[x,y]\!\!\! &=& 0, 
\label{identity 4}\\
2\,u\cdot v\cdot[x,y]=(-1)^{|x|\,|v|}[u\cdot x,v\cdot y]+(-1)^{|u|\,(|x|+|v|)}[v\cdot x,u\cdot y],
\label{identity 5}\\
(-1)^{|u|\,|yv|}\,x\cdot[u,y\cdot v]+(-1)^{|v|\,|xy|}\,v\cdot[x\cdot y,u]+(-1)^{|x|\,|yv|}y\cdot[v,x]\cdot u\!\!\! &=&\!\!\! 0.
\label{identity 6}
\end{eqnarray}
\label{Theorem 2}
\end{theorem}
Another notion that will play an important role in this paper is a notion of 3-Lie superalgebra. 3-Lie superalgebra is an extension of a notion of 3-Lie algebra to the case of superalgebras. A notion of 3-Lie algebra is a particular case of a notion of $n$-Lie algebra ($n\geq 2$) proposed and developed by Filippov \cite{Filippov}. Let $L$ be a $\mathbb K$-vector space. Then $(L,[\;,\;,\;])$ is said to be a 3-Lie algebra if ternary bracket $[\;,\;,\;]:L\times L\times L\to L$ is totally skew-symmetric
\begin{equation*}
    [x,y,z] = -[y,x,z],\;\;[x,y,z] = -[x,z,y],\;\;\;x,y,z\in L,
\end{equation*}
and any five elements $x,y,z,u,v\in L$ satisfy the Filippov-Jacobi identity
\begin{equation*}
    [[x,y,z],u,v] = [[x,u,v],y,z] + [[y,u,v],z,x] + [[z,u,v],x,y].
\end{equation*}
\begin{defn}
    Let ${\cal L}={\cal L}_0\oplus {\cal L}_1$ be a super $\mathbb K$-vector space. $(\cal L,[\,,\,,\,])$ is said to be a 3-Lie superalgebra, if a $\mathbb K$-trilinear ternary bracket $[\,,\,,\,]:{\cal L}^3\to \cal L$ has the following properties
    \begin{equation}\label{parity-compatibility}
        |[x,y,z]| = |x|+|y|+|z|,
    \end{equation}
    \begin{equation}\label{ternary-skew-symmetry}
        [y,x,z] = -(-1)^{|x||y|}[x,y,z],\quad [x,z,y] = -(-1)^{|y||z|}[x,y,z],
    \end{equation}
and satisfies the super Filippov-Jacobi identity
    \begin{equation}\label{Filippov-Jacobi}
        \begin{split}
        [[x,y,z],u,v] = (-1)^{|yz,uv|}[[x,u,v],y,z]+ (-1)^{|x,yz|+|xz,uv|}[[y,u,v],z,x] +(-1)^{|xy,zuv|}[[z,u,v],x,y],
        \end{split}
    \end{equation}
where $|yz,uv|=(|y|+|z|)(|u|+|v|),\;|x,yz|=|x|(|y|+|z|),\; |xy,zuv|=(|x|+|y|)(|z|+|u|+|v|)$.
\end{defn}
%%%%%%%%%%%%%%%%%%%%%%%%%%%%%%%%%%%%%
%%%%%%%%%%%%%%%%%%%%%%%%%%%%%%%%%%%%%
%%%%%%%%%%%%%%%%%%%%%%%%%%%%%%%%%%%%%
\section{The case of odd derivation}
%%%%%%%%%%%%%%%%%%%%%%%%%%%%%%%%%%%%%%
%%%%%%%%%%%%%%%%%%%%%%%%%%%%%%%%%%%%%
%%%%%%%%%%%%%%%%%%%%%%%%%%%%%%%%%%%%%%%
A method for constructing a transposed Poisson algebra if we are given a commutative associative algebra and its derivation was proposed in \cite{Bai-Guo-Wu}. In \cite{Abramov 1} it was proved that this method can also be applied to construct a transposed Poisson superalgebra, that is, if we are given a commutative associative superalgebra and its even derivation, then we can construct a transposed Poisson superalgebra by equipping commutative superalgebra with the graded Lie bracket (\ref{graded bracket with even derivation}). In this section we consider the question of how this method can be extended to odd derivations of commutative superalgebra. This question is important because odd derivations play a significant role not only in the structure of a commutative superalgebra but also in applications. In this case, odd derivations whose square is zero play a particularly important role. For example, the exterior differential in a graded differential  algebra of differential forms and the BRST-operator in quantum field theory \cite{Slavnov-Faddeev}.

Thus, our goal in this section is to include odd derivations of a commutative superalgebra in the scheme of constructing algebras like a transposed Poisson superalgebra. For this purpose we introduce the following notion.
\begin{defn}
Let ${\cal E}={\cal E}_0\oplus{\cal E}_1$ be a left supermodule over a commutative superalgebra $A=A_0\oplus A_1$, that is, $(x,X)\in A\times{\cal E}\mapsto x\cdot X\in {\cal E}$, where $|x\cdot X|=|x|+|X|$ and $|x|,|X|$ are degrees (mod 2). Let $\{\;,\;\}:{\cal E}\times {\cal E}\to {\cal E}$ be a $\mathbb K$-bilinear mapping such that
\begin{itemize}
\item[a)] $|\{X,Y\}|=|X|+|Y|$,
\item[b)] $\{X,Y\}=(-1)^{|X||Y|}\{Y,X\}$,
\item[c)] $({\cal E}_0,\{\;,\;\})$ is a Jordan algebra,
\item[d)] the mappings $\{\;,\;\}:{\cal E}_0\times{\cal E}_1\to{\cal E}_1$ and $\{\;,\;\}:{\cal E}_1\times{\cal E}_0\to{\cal E}_1$ define the Jordan ${\cal E}_1$ over ${\cal E}_0$.
\end{itemize}
Then we say that a left $A$-supermodule $\cal E$ is transposed Poisson type compatible with a Jordan structure of $\cal E$ if for any $z\in A, X,Y\in{\cal E}$ we have the identity
$$
2\,z\cdot\{X,Y\}=\{z\cdot X,Y\}+(-1)^{|z|\,|X|}\{X,z\cdot Y\}.
$$
\end{defn}
\noindent

From this definition it follows that if $X_0\in{\cal E}_0, X_1\in{\cal E}_1$ then $\{X_0,X_1\}=\{X_1,X_0\}$. This shows that the elements of ${\cal E}_1$ obtained from $X_1$ by the left and right actions of $X_0\in{\cal E}_0$ are equal as it should be in the case of a Jordan module. Thus the only condition which should be checked is that ${\cal E}={\cal E}_0\oplus{\cal E}_1$ is a Jordan algebra if we equip it with the multiplication
\begin{equation}
(X_0,X_1).(Y_0,Y_1)=(\{X_0,Y_0\},\{X_0,Y_1\}+\{X_1,Y_0\})
\end{equation}
In what follows we assume that $A$ is a commutative superalgebra and multiplication in this algebra will be denoted by means of a juxtaposition, that is,  writing one element after another. Then the vector space of derivations of this algebra is a Lie superalgebra if we equip it with the bracket
\begin{equation}
[D_1,D_2]=D_1\circ D_2-(-1)^{|D_1||D_2|}\,D_2\circ D_1,
\label{initial graded commutator}
\end{equation}
where $D_1,D_2$ are derivations of $A$ and $|D_1|,|D_2|$ are their parities. Now assume that $D$ is a derivation of a superalgebra $A$ and $D$ has a certain parity, that is, $D$ is either even or odd derivation. We consider the super vector space ${\mathfrak D}^D=\{x\cdot D:x\in A\}$ generated by $D$. Here $x\cdot D$ is a derivation of a superalgebra $A$ defined by $(x\cdot D)(y)=x D(y)$ and $|x\cdot D|=|x|+|D|$. Using the terminology of differential geometry, we will call derivations $x\cdot D$ vector fields. Obviously ${\mathfrak D}^D$ becomes a left $A$-supermodule if one defines $x\cdot(y\cdot D)=(xy)\cdot D$. Let $X=x\cdot D, Y=y\cdot D$ and calculating the bracket (\ref{initial graded commutator}) we find
\begin{eqnarray}
[X,Y] \!\!&=&\!\!\big(x\,D(y)-(-1)^{(|x|+|D|)(|y|+|D|)}y\,D(x)\big)\cdot D\nonumber\\
      &&\;\;\;\;\;\;\;+\big((-1)^{|y||D|}-(-1)^{|y||D|+|D|^2}\big)\, (xy)\cdot D^2.\label{graded commutator}
\end{eqnarray}
Assume $D$ is an even derivation of $A$, that is, $|D|=0$. In this case the second term (containing $D^2$) at the right-hand side of the above formula vanishes and we get
\begin{equation}
[X,Y] =\big(x\,D(y)-(-1)^{|x||y|}y\, D(x)\big)\cdot D.
\label{bracket with even derivation}
\end{equation}
Thus, the super vector space ${\mathfrak D}^D$ generated by an even degree derivation $D$ closes with respect to the graded commutator (\ref{graded commutator}) and we obtain the structure of the Lie superalgebra on ${\mathfrak D}^D$. It is easy to see that we can omit $D$ in both sides of (\ref{bracket with even derivation}) and the consistency of parities will not be broken. Figuratively speaking, we can descend the bracket (\ref{bracket with even derivation}) from vector fields ${\mathfrak D}^D$ to elements of a superalgebra $A$. In this case we get
\begin{equation}
[x,y] =x\,D(y)-(-1)^{|x||y|}y\,D(x).
\label{even derivation bracket on elements}
\end{equation}
Then it can be proved \cite{Abramov 1} that the bracket (\ref{even derivation bracket on elements}) satisfies the compatibility condition
$$
2\,z\,[x,y]=[z x,y]+(-1)^{|z|\,|x|}[x,z y].
$$
Hence a commutative superalgebra $A$ endowed with the bracket (\ref{even derivation bracket on elements}) is a transposed Poisson superalgebra.

Now we consider the case of odd derivation of a commutative superalgebra $A$. This odd derivation will be denoted by $\delta$, that is, $|\delta|=1$. Let us consider the super vector space ${\mathfrak D}^\delta=\{x\cdot \delta: x\in A\}$ induced by an odd derivation $\delta$. In the case of an odd derivation we have $|y||\delta|=|y|,\;|y||\delta|+|\delta|^2=|y|+1$ and the second term on the right side of (\ref{graded commutator}), containing $\delta^2$, does not vanish. However, if in the right side of formula (\ref{graded commutator}) we take plus instead of minus, that is, we consider the bracket
\begin{eqnarray}
\{X,Y\} \!\!&=&\!\!\big(x\,\delta(y)+(-1)^{(|x|+|\delta|)(|y|+|\delta|)}y\,\delta(x)\big)\cdot \delta\nonumber\\
      &&\;\;\;\;\;\;\;+\big((-1)^{|y|}+(-1)^{|y|+1}\big)\, (xy)\cdot \delta^2.\label{}
\end{eqnarray}
then the term containing $\delta^2$ vanishes. For $X=x\cdot\delta$ and $Y=y\cdot\delta$ we will have
\begin{equation}
\{X,Y\} =\big(x\,\delta(y)+(-1)^{(|x|+1)(|y|+1)}y\, \delta(x)\big)\cdot \delta.
\label{graded commutator odd derivation}
\end{equation}
Comparing with the case of even derivation we see that in the case of odd derivation $\delta$ we can not omit $\delta$ in vector fields in (\ref{graded commutator odd derivation}) because we will break the consistency of parities, that is, we will get $|\{x,y\}|=|x|+|y|+1$ instead of what it should be 
$|\{x,y\}|=|x|+|y|$. Hence we can not, so to speak, descend the bracket (\ref{graded commutator odd derivation}) from the vector fields ${\mathfrak D}^\delta$ on to a superalgebra $A$.

Obviously ${\mathfrak D}^\delta$ becomes a super vector space if one defines the parity of a vector field $x\cdot\delta$ as $|x|+1$. The super vector space ${\mathfrak D}^\delta$  is a left $A$-supermodule if we define $y\cdot (x\cdot \delta)=(y x)\cdot\delta$. 
\begin{theorem}
$({\mathfrak D}^\delta_0,\{\;,\;\})$ is a Jordan algebra. If $\delta^2=0$ then the bracket (\ref{graded commutator odd derivation}) defines on ${\mathfrak D}^\delta_1$ the structure of a Jordan module over ${\mathfrak D}^\delta_0$ . The left $A$-supermodule ${\mathfrak D}^\delta$ is transposed Poisson type compatible with the Jordan structure of ${\mathfrak D}^\delta$, that is, we have the identity
\begin{equation}
2 z\cdot \{X,Y\}=\{z\cdot X,Y\}+(-1)^{|z||X|}\{X,z\cdot Y\},
\label{compatibility odd derivation}
\end{equation}
\end{theorem}
\begin{proof}
First of all, we prove the transposed Poisson type compatibility condition (\ref{compatibility odd derivation}). Applying (\ref{graded commutator odd derivation}) to the left side of (\ref{compatibility odd derivation}) we find
\begin{equation}
2 z\cdot \{X,Y\}=2\,\big(z\, x\,\delta(y)+(-1)^{|X||Y|}\,z\,y\,\delta(x)\big)\cdot \delta.
\label{proposition 1,1}
\end{equation}
Calculating the terms at the right side of (\ref{compatibility odd derivation}) we get
\begin{eqnarray}
\{z\cdot X,Y\}\!\!\! &=&\!\!\!\big(z\,x\,\delta(y) +(-1)^{(|z|+|X|) |Y|}y\,\delta(z)\,x+(-1)^{|X||Y|}z\,y\, 
      \delta(x)\big)\cdot\delta,\nonumber\\
\left\{X,z\cdot Y\right\}\!\!\! &=&\!\!\! \big(x\,\delta(z)\,y +(-1)^{|z||X|}z\,x\,\delta(y)+(-1)^{|X|(|z|+|Y|)}z\,y\, 
      \delta(x)\big)\cdot \delta.\nonumber
\end{eqnarray}
Now multiplying the both sides of the second equation by $(-1)^{|z||X|}$ and taking the sum with the first equation we get the right side of (\ref{proposition 1,1}).

A proof that $({\mathfrak D}^\delta_0,\{\;,\;\}$ is a Jordan algebra can be found in \cite{Roger}. Since the proof given in \cite{Roger} is very short and, in our opinion, does not explain some details, we give a more detailed proof here. It is easy to see that the bracket (\ref{graded commutator odd derivation}) has the symmetry
\begin{equation}
\{X,Y\}=(-1)^{|X||Y|}\{Y,X\}.
\label{symmetry of bracket}
\end{equation}
Hence if $X=x\cdot \delta,Y=y\cdot\delta\in{\mathfrak D}^\delta_0$ then $|x|=|y|=1$ and $\{X,Y\}=\{Y,X\}$, i. e. the bracket (\ref{graded commutator odd derivation}) restricted to the subspace ${\mathfrak D}^\delta_0$ is commutative. Hence, in order to prove that $({\mathfrak D}^\delta_0,\{\;,\;\})$ is a Jordan algebra we need to prove the Jordan identity
\begin{equation}
\big\{\{X,X\},\{Y,X\}\big\}=\big\{\{\{X,X\},Y\},X\big\}.
\label{Jordan identity}
\end{equation}
By straightforward calculation we find
\begin{eqnarray}
\big\{\{X,X\},\{Y,X\}\big\} \!\!\!\!&=&\!\!\!\! \big(6\,x\,(\delta(x))^2\,\delta(y)-2\,x^2\,\delta(x)\,\delta^2(y)+2\,y\,(\delta(x))^3-2\,x^2\,\delta^2(x)\,\delta(y)\big)\cdot\delta,\label{calculations 1}\\
\big\{\{\{X,X\},Y\},X\big\} \!\!\!\!&=&\!\!\!\! \big(6\,x\,(\delta(x))^2\,\delta(y)-2\,x^2\,\delta(x)\,\delta^2(y)+2\,y\,(\delta(x))^3-4\,x^2\,\delta^2(x)\,\delta(y)+\nonumber\\
&&\qquad\qquad\qquad\qquad\qquad\qquad\qquad\qquad\qquad\qquad+2\,x^2\,y\,\delta^3(x)\big)\cdot\delta.\label{calculations 2}
\end{eqnarray}
Since a superalgebra $A$ is commutative for any odd element $x\in A_1$ we have $x^2=0$. Therefore, all terms on the right-hand sides of the above equalities containing $x^2$ vanish and we see that the right-hand sides are equal. From this follows the Jordan identity.

Let us now prove that if $\delta^2=0$ then ${\mathfrak D}^\delta_1$ is a Jordan module over the Jordan algebra ${\mathfrak D}^\delta_0$. First of all, we have to show that the left and right actions of ${\mathfrak D}^\delta_0$ on ${\mathfrak D}^\delta_1$ determined by the bracket (\ref{graded commutator odd derivation}) coincide, that is, for any $X\in {\mathfrak D}^\delta_0$ and $Y\in {\mathfrak D}^\delta_1$ it holds $\{X,Y\}=\{Y,X\}$. But this follows immediately from formula (\ref{symmetry of bracket}) when $X\in{\mathfrak D}^\delta_0,Y\in{\mathfrak D}^\delta_1$. Then we must show that (\ref{Jordan identity}) holds when $X\in {\mathfrak D}^\delta_0$ and $Y\in{\mathfrak D}^\delta_1$. But in this case formulae (\ref{calculations 1}), (\ref{calculations 2}) do not change if we assume that $X$ is an even vector field and $Y$ is odd. Indeed, if $X$ is an even vector field and $Y$ is odd, the formula for the bracket (\ref{graded commutator odd derivation}) will have the same form as in the case of even vector fields, that is, $\{X,Y\}=(x\,\delta(y)+y\,\delta(x))\cdot\delta$.

The only identity which we need to prove in order to show that ${\mathfrak D}^\delta_1$ is a Jordan module over the Jordan algebra ${\mathfrak D}^\delta_0$ is the following
\begin{equation}
\Big\{\big\{\{X,X\},Y\big\},Z\Big\}-\Big\{\{X,X\},\{Y,Z\}\Big\}=2\,\Big\{\{X,Y\},\{X,Z\}\Big\}-2\,\Big\{X,\big\{Y,\{X,Z\}\big\}\Big\}.
\label{Jordan module identity}
\end{equation}
Now we assume $\delta^2=0$. We have
\begin{eqnarray}
\{X,X\} &=& 2\,(x\,\delta(x))\cdot\delta,\nonumber\\
\Big\{\big\{\{X,X\},Y\big\},Z\Big\} &=& \big(2\,x\,\delta(x)\,\delta(y)\,\delta(z)+2\,(\delta(x))^2\,y\,\delta(z)+4\,(\delta(x))^2\,\delta(y)\,z\big)\cdot\delta,\nonumber\\
\Big\{\{X,X\},\{Y,Z\}\Big\} &=& \big(4\,x\,\delta(x)\,\delta(y)\,\delta(z)+2\,(\delta(x))^2\,y\,\delta(z)+2\,(\delta(x))^2\,\delta(y)\,z\big)\cdot\delta,\nonumber
\end{eqnarray}
Hence the left-hand side of (\ref{Jordan module identity}) is equal to
\begin{equation}
-2\,x\,\delta(x)\,\delta(y)\,\delta(z)+2\,(\delta(x))^2\,\delta(y)\,z.
\label{Jordan module identity left side}
\end{equation}
The terms at the right-hand side of (\ref{Jordan module identity}) can be written as follows
\begin{eqnarray}
2\,\Big\{\{X,Y\},\{X,Z\}\Big\} &=& \big(8\,x\,\delta(x)\,\delta(y)\,\delta(z)+4\,(\delta(x))^2\,y\,\delta(z)+4\,(\delta(x))^2\,\delta(y)\,z\big)\cdot\delta,\nonumber\\
2\,\Big\{X,\big\{Y,\{X,Z\}\big\}\Big\} &=& \big(10\,x\,\delta(x)\,\delta(y)\,\delta(z)+4\,(\delta(x))^2\,y\,\delta(z)+2\,(\delta(x))^2\,\delta(y)\,z\big)\cdot\delta.\nonumber
\end{eqnarray}
Hence the right-hand side of (\ref{Jordan module identity}) is the expression
\begin{equation}
-2\,x\,\delta(x)\,\delta(y)\,\delta(z)+2\,(\delta(x))^2\,\delta(y)\,z,\nonumber
\end{equation}
and comparing it with the expression at the left-hand side (\ref{Jordan module identity left side}) we conclude that the Jordan module identity is satisfied and this ends the proof. 
\end{proof}
%%%%%%%%%%%%%%%%%%%%%%%%%%%%%%%%%%%%%%%
%%%%%%%%%%%%%%%%%%%%%%%%%%%%%%%%%%%%%%%
\section{3-Lie superalgebra constructed by means of even derivation}
%-----------------
%-----------------

In the article \cite{Bai-Guo-Wu} the authors prove an important theorem which gives a method for constructing 3-Lie algebras by means of transposed Poisson algebras and their derivations. This theorem can be stated as follows
\begin{theorem}
    Let $(L,\cdot,[\,,\,])$ be a transposed Poisson algebra and $D$ be its derivation. Define a ternary operation on $L$ as follows
    \begin{equation}\label{Bai-ternary-bracket}
    [x,y,z]:=D(x)[y,z]+D(y)[z,x]+D(z)[x,y], \quad x,y,z\in L.
    \end{equation}
    Then $(L,[\,,\,,\,])$ is a 3-Lie algebra.
\end{theorem}
\noindent
Our aim in this section is to extend this result to the case of transposed Poisson superalgebra. We show that given a transposed Poisson superalgebra and an even derivation of this superalgebra we can construct a 3-Lie superalgebra. 

We will start with the following lemma. 

\begin{lemma}\label{lem}
If $({\cal P},\cdot,[\,,\,])$ is a transposed Poisson superalgebra and $D$ is an even derivation of a Lie superalgebra $({\cal P},[\,,\,])$ then 
\begin{eqnarray}
&&\!\!\!\!\!\!\!\!\!\!\! D(x)\!\cdot\! D([y,z])+(-1)^{|x,yz|}D(y)\!\cdot\! D([z,x])+(-1)^{|xy,z|}D(z)\!\cdot\! D([x,y])\nonumber\\ 
&&\!\!\!\!\!\!\!\! +x\!\cdot\! [D(y),D(z)]+(-1)^{|x,yz|} y\!\cdot\! [D(z),D(x)])+(-1)^{|xy,z|} z\!\cdot\! [D(x),D(y)]=0.\nonumber
\end{eqnarray}
\end{lemma}

\begin{proof}
Since $D:{\cal P}\to {\cal P}$ is an even derivation of a Lie superalgebra $({\cal P},[\,,\,])$, it satisfies the graded Leibniz rule. Since it is an even derivation we also have $|D(x)|=|x|$ for any $x\in {\cal P}$. Making use of the graded Leibniz rule, the compatibility condition (\ref{comp-super-transposed}) and cyclic permutations of $x,y,z$ we get
\begin{eqnarray}
D(x)\cdot D([y,z]) = \frac{1}{2}\big{(}[D(x)\cdot D(y),z])+(-1)^{|x||y|}[D(y),D(x)\cdot z]
+[D(x)\cdot y,D(z)]+(-1)^{|x||y|}[y,D(x)\cdot D(z)]\big{)}\nonumber,\\
D(y)\cdot D([z,x])=
     \frac{1}{2}\big{(}[D(y)\cdot D(z),x])+(-1)^{|y||z|}[D(z),D(y)\cdot x]
      +[D(y)\cdot z,D(x)]+(-1)^{|y||z|}[z,D(y)\cdot D(x)]\big{)},\nonumber \\
D(z)\cdot D([x,y])=
     \frac{1}{2}\big{(}[D(z)\cdot D(x),y])+(-1)^{|x||z|}[D(x),D(z)\cdot y]
      +[D(z)\cdot x,D(y)]+(-1)^{|x||z|}[x,D(z)\cdot D(y)]\big{)}.\nonumber
\end{eqnarray}
Taking the sum of these equations multiplied by $(-1)^{|x,yz|}$ (second equation) and $(-1)^{|xy,z|}$ (third equation) we get the equation whose left-hand side is
\begin{equation}
D(x)\cdot D([y,z])+(-1)^{|x,yz|}D(y)\cdot D([z,x])+(-1)^{|xy,z|}D(z)\cdot D([x,y]),
\end{equation}
and the right-hand side can be written in the form
\begin{eqnarray}
&&-\frac{1}{2}\big([x\cdot D(y),D(z)]+(-1)^{|x||y|}[D(y),x\cdot D(z)]+(-1)^{|x,yz|}([y\cdot D(z),D(x)]+(-1)^{|y||z|}[D(z),y\cdot D(x)]\nonumber\\
    &&\qquad\qquad\qquad\;\;\;\;\;\;\;\;\;\;\;\;\;+(-1)^{|xy,z|}([z\cdot D(x),D(y)]+(-1)^{|x||z|}[D(x),z\cdot D(y)])\big).\nonumber
\end{eqnarray}
Making use of the compatibility condition for transposed Poisson superalgebra we can write the right-hand side in the form
$$
-x\cdot [D(y),D(z)]-(-1)^{|x,yz|} y\cdot [D(z),D(x)])-(-1)^{|xy,z|} z\cdot [D(x),D(y)],
$$
which ends the proof of lemma.
\end{proof}

\begin{theorem}
Let $({\cal P},\cdot,[\,,\,])$ be a transposed Poisson superalgebra and $D$ be its even derivation. Define the ternary bracket
\begin{equation}\label{ternary bracket with D}
[x,y,z]:=D(x)\cdot [y,z]+(-1)^{|x,yz|}D(y)\cdot[z,x]+(-1)^{|xy,z|}D(z)\cdot[x,y],
\end{equation}
where $x,y,z\in {\cal P}.$ Then $({\cal P},[\,,\,,\,])$ is a 3-Lie superalgebra.
\end{theorem}

\begin{proof}
It is easy to verify that the ternary bracket (\ref{ternary bracket with D}) is trilinear and graded skew-symmetric. Hence in order to prove the theorem we have to prove the super Filippov-Jacobi identity (\ref{Filippov-Jacobi}). We begin with the left-hand side of (\ref{Filippov-Jacobi}). Applying the definition of ternary bracket (\ref{ternary bracket with D}) we get
\begin{equation*}
    \begin{split}
        [&[x,y,z],u,v] =  D\Big(D(x)[y,z]\Big)[u,v]+(-1)^{|xyz,uv|}D(u)[v,D(x)[y,z]]+(-1)^{|v,xyzu|}D(v)[D(x)[y,z],u] \\
        &\quad\quad\quad+ (-1)^{|x,yz|} D\Big(D(y)[z,x]\Big)[u,v]+(-1)^{|xyz,uv|+|x,yz|}D(u)[v,D(y)[z,x]] \\
        &\quad\quad\quad+(-1)^{|v,xyzu|+|x,yz|}D(v)[D(y)[z,x],u]  
        + (-1)^{|xy,z|} D\Big(D(z)[x,y]\Big)[u,v] \\ 
        &\quad\quad\quad+(-1)^{|xyz,uv|+|xy,z|}D(u)[v,D(z)[x,y]]+(-1)^{|v,xyzu|+|xy,z|}D(v)[D(z)[x,y],u].
    \end{split}
\end{equation*}
Making use of the graded Leibniz rule for even derivation $D$ we obtain 
\begin{equation*}
    \begin{split}
    D\Big(D(x)[y,z]\Big)[u,v] &= D^2(x)[y,z][u,v]+\underline{D(x)D([y,z])}[u,v], \\
    (-1)^{|x,yz|}D\Big(D(y)[z,x]\Big)[u,v] &= (-1)^{|x,yz|}D^2(y)[z,x][u,v]+\underline{(-1)^{|x,yz|}D(y)D([z,x])}[u,v], \\
    (-1)^{|xy,z|}D\Big(D(z)[x,y]\Big)[u,v] &= (-1)^{|xy,z|}D^2(z)[x,y][u,v]+\underline{(-1)^{|xy,z|}D(z)D([x,y])}[u,v].
    \end{split}
\end{equation*}
Note that the sum of underlined terms is equal to the sum of the first three terms in equation in Lemma \ref{lem} multiplied by $[u,v]$. Hence using the identity of Lemma \ref{lem} we can write this sum as follows
\begin{equation*}
\begin{split}
    D(x) D([y,z])[u,v]+(-1)^{|x,yz|}D(y)D([z,x])[u,v]+(-1)^{|xy,z|}D(z)D([x,y])[u,v] &=\\ 
-x [D(y),D(z)][u,v]-(-1)^{|x,yz|} y [D(z),D(x)][u,v]-(-1)^{|xy,z|} z[D(x),&D(y)][u,v].
\end{split}
\end{equation*}
Now the left-hand side of super Filippov-Jacobi identity can be written in the following form
\begin{equation}\label{left-side}
    \begin{split}
         & D^2(x)[y,z][u,v]+ (-1)^{|x,yz|}D^2(y)[z,x][u,v] + (-1)^{|xy,z|} D^2(z)[x,y][u,v] \\
        &  + (-1)^{|xyz,uv|}D(u)\big( [v,D(x)[y,z]]+(-1)^{|x,yz|}[v,D(y)[z,x]] + (-1)^{|xy,z|}[v,D(z)[x,y]] \big) \\
        &+ (-1)^{|xyzu,v|}D(v)\big( [D(x)[y,z],u] + (-1)^{|x,yz|}[D(y)[z,x], u] + (-1)^{|xy,z|}[D(z)[x,y],u] \big) \\
        &- x[D(y),D(z)][u,v] - (-1)^{|x,yz|}y[D(z),D(x)][u,v] - (-1)^{|xy,z|}z[D(x),D(y)][u,v]. 
    \end{split}
\end{equation}
We label every term of this expression by a pair (L, $n$), where L stands for the "left-hand side of the super Filippov-Jacobi identity" and $n$ is the positional number of the term in this expression. For example, (L, 2) is a label for $(-1)^{|x,yz|}D^2(y)[z,x][u,v]$. 

Now we calculate the terms at the right-hand side of the super Filippov-Jacobi identity (\ref{Filippov-Jacobi}).  We get

\begin{equation}\label{right-hand-1}
    \begin{split}
        (-1)^{|yz,uv|}&[[x,u,v],y,z]= D^2(x)[y,z][u,v]-(-1)^{|yz,uv|}x[D(u),D(v)][y,z] \\
    &+(-1)^{|x,yz|}D(y)[z,D(x)[u,v]] +(-1)^{|xy,z|+|y,uv|}D(z)[D(x)[u,v],y] \\
    &+ (-1)^{|xyz,uv|}D^2(u)[v,x][y,z]- (-1)^{|xyz,uv|}u [D(v),D(x)])[y,z] \\
    &+(-1)^{|x,yzuv|}D(y)[z,D(u)[v,x]]+(-1)^{|xy,zuv|}D(z)[D(u)[v,x],y] \\
    &+ (-1)^{|xu,v|+|yz,uv|} D^2(v)[x,u][y,z]-(-1)^{|xu,v|+|yz,uv|}v[D(x),D(u)][y,z] \\
    &+(-1)^{|x,yz|+|xu,v|}D(y)[z,D(v)[x,u]] +(-1)^{|xy,z|+|xu,v|+|y,uv|}D(z)[D(v)[x,u],y].
    \end{split}
\end{equation}

\begin{equation}\label{right-hand-2}
    \begin{split}
        (-1)^{|x,yz|+|xz,uv|}[&[y,u,v],z,x]= (-1)^{|x,yz|}D^2(y)[z,x][u,v] -(-1)^{|x,yz|+|xz,uv|}y[D(u),D(v)][z,x] \\
    &+(-1)^{|xy,z|}D(z)[x,D(y)[u,v]] +(-1)^{|z,uv|}D(x)[D(y)[u,v],z] \\ &+ (-1)^{|xyz,uv|+|x,yz|}D^2(u)[v,y][z,x] - (-1)^{|xyz,uv|+|x,yz|}u [D(v),D(y)])[z,x] \\ &+(-1)^{|xy,z|+|y,uv|}D(z)[x,D(u)[v,y]] +(-1)^{|yz,uv|}D(x)[D(u)[v,y],z] \\ &+ (-1)^{|x,yz|+|u,xzv|+|xyz,v|} D^2(v)[y,u][z,x] -(-1)^{|x,yz|+|u,xzv|+|xyz,v|}v[D(y),D(u)][z,x] \\ &+(-1)^{|xy,z|+|yu,v|}D(z)[x,D(v)[y,u]] +(-1)^{|z,uv|+|yu,v|}D(x)[D(v)[y,u],z].
    \end{split}
\end{equation}

\begin{equation}\label{right-hand-3}
    \begin{split}
        (-1)^{|xy,zuv|}&[[z,u,v],x,y]= (-1)^{|xy,z|}D^2(z)[x,y][u,v]-(-1)^{|xy,zuv|}z[D(u),D(v)][x,y] \\
    &+D(x)[y,D(z)[u,v]] +(-1)^{|x,yzuv|}D(y)[D(z)[u,v],x] \\
    &+ (-1)^{|xyz,uv|+|xy,z|}D^2(u)[v,z][x,y]- (-1)^{|xyz,uv|+|xy,z|}u [D(v),D(z)])[x,y] \\
    &+(-1)^{|z,uv|}D(x)[y,D(u)[v,z]]+(-1)^{|z,xuv|+|x,yuv|}D(y)[D(u)[v,z],x] \\
    &+ (-1)^{|zu,xyv|+|xy,v|} D^2(v)[z,u][x,y]-(-1)^{|zu,xyv|+|xy,v|}v[D(z),D(u)][x,y] \\
    &+(-1)^{|zu,v|}D(x)[y,D(v)[z,u]]+(-1)^{|x,yzuv|+|zu,v|}D(y)[D(v)[z,u],x].
    \end{split}
\end{equation}

Analogously we label terms of these expressions by pairs $(m,n)$, where $m$ is a Roman number I, II or III, where I stands for (\ref{right-hand-1}), II for (\ref{right-hand-2}) and III for (\ref{right-hand-3}), and $n$ is the positional number of a term in corresponding expression. For example, the pair (I, 1) denotes the term $(-1)^{|yz,uv|}D^2(x)[u,v][y,z]$. We see that the terms (I,1) and (L, 1), (II,1) and (L,2), (III,1) and (L,3) cancel. Our next goal is to show that the remaining terms containing a square of a derivation $D$ can be also cancelled. First collect all the terms with $D^2(u)$
\begin{equation*}
    \begin{split}
         (-1)^{|xyz,uv|}D^2(u)\Big( [v,x][y,z]+(-1)^{|x,yz|}[v,y][z,x]+(-1)^{|z,xy|}[v,z][x,y] \Big).
    \end{split}
\end{equation*}
This sum is equal to zero due to identity (\ref{identity 4}). Analogously all the terms with $D^2(v)$ can be canceled. The sum of the terms containing $[D(u),D(v)]$ is equal to zero due to identity (\ref{identity 1}). Now making use of identity (\ref{identity 2}) we obtain
\begin{equation*}
    \begin{split}
        \text{(III,12)}+\text{(I,11)}=-(-1)^{|xzu,v|+|x||y|}D(y)[D(v)[x,z],u], \; 
        \text{(II,7)}+\text{(I,8)}=-(-1)^{|xy,zu|}D(z)[D(u)[x,y],v],
    \end{split}
\end{equation*}
\begin{equation*}
    \begin{split}
        \text{(II,12)}+\text{(III,11)}=-(-1)^{|yzu,v|+|y||z|}D(x)[D(v)[z,y],u], \; \text{(I,7)}+\text{(III,8)}=-(-1)^{|x,yzu|+|z||u|}D(y)[D(u)[z,x],v],
    \end{split}
\end{equation*}
\begin{equation*}
    \begin{split}
        \text{(I,12)}+\text{(II,11)}=-(-1)^{|x,yzv|+|y,zv|+|u||v|}D(z)[D(v)[y,x],u], \; \text{(III,7)}+\text{(II,8)} = -(-1)^{|yz,u|}D(x)[D(u)[y,z],v].
    \end{split}
\end{equation*}

After all cancellations at the right-hand side of the super Filippov-Jacobi identity we have the following terms 
\begin{equation}\label{RHS}
    \begin{split}
         \underline{-(-1)^{|yz,u|}D(x)[D(u)[y,z],v]} -\dashuline{(-1)^{|x,yzu|+|z||u|}D(y)[D(u)[z,x],v]}& \\
         \mathdotuline{black}{-(-1)^{|xy,zu|}D(z)[D(u)[x,y],v]}& \\
         \mathunderline{red}{-(-1)^{|yzu,v|+|y||z|}D(x)[D(v)[z,y],u]}-\mathdashuline{red}{(-1)^{|xzu,v|+|x||y|}D(y)[D(v)[x,z],u]}& \\
         \mathdotuline{red}{-(-1)^{|x,yzv|+|y,zv|+|u||v|}D(z)[D(v)[y,x],u]}& \\
         +D(x)[y,D(z)[u,v]] + (-1)^{|x,yz|}D(y)[z,D(x)[u,v]]& \\
         \mathdotuline{white}{+ (-1)^{|xy,z|}D(z)[x,D(y)[u,v]]}& \\
         +(-1)^{|z,uv|}D(x)[D(y)[u,v],z]+ (-1)^{|x,yzuv|}D(y)[D(z)[u,v],x]& \\ + \mathdotuline{white}{(-1)^{|xy,z|+|y,uv|}D(z)[D(x)[u,v],y]}& \\
        \mathunderline{red}{-(-1)^{|xyz,uv|}u\big( [D(v),D(x)][y,z]} +\mathdashuline{red}{(-1)^{|x,yz|}[D(v),D(y)][z,x]}& \\ \mathdotuline{red}{+ (-1)^{|xy,z|}[D(v),D(z)][x,y]}& \big) \\
         +\underline{(-1)^{|xyz,uv|+ |u||v|}v\big( [D(u),D(x)][y,z]} + \mathdashuline{black}{(-1)^{|x,yz|}[D(u),D(y)][z,x]}& \\ + \mathdotuline{black}{(-1)^{|xy,z|}[D(u), D(z)][x,y]}& \big).
    \end{split}
\end{equation}
The left-hand side has the form
\setulcolor{red}
\begin{equation}\label{LHS}
    \begin{split}
        & \underline{(-1)^{|xyz,uv|}D(u)\big( [v,D(x)[y,z]]}+\dashuline{(-1)^{|x,yz|}[v,D(y)[z,x]]} + \dotuline{(-1)^{|xy,z|}[v,D(z)[x,y]]} \big) \\
        &+ \mathunderline{red}{(-1)^{|xyzu,v|}D(v)\big([D(x)[y,z],u]} + \mathdashuline{red}{(-1)^{|x,yz|}[D(y)[z,x], u]} + \mathdotuline{red}{(-1)^{|xy,z|}[D(z)[x,y],u]} \big) \\
        &- x[D(y),D(z)][u,v] - (-1)^{|x,yz|}y[D(z),D(x)][u,v] - (-1)^{|xy,z|}z[D(x),D(y)][u,v]. 
    \end{split}
\end{equation}
Making the following substitution $x\to D(u)$, $u\to v$, $y\to [y,z]$, $v\to D(x)$ in identity (\ref{identity 6}), we get 
\begin{equation}
    \begin{split}
        (-1)^{|xyz,uv|}D(u) [v,D(x)[y,z]] = -&(-1)^{|yz,u|}D(x)[D(u)[y,z],v] \\ 
        &+ (-1)^{|xyz,uv|+ |u||v|}v\,[D(u),D(x)][y,z]
    \end{split}
\end{equation}
In this equation the left-hand side is equal to the first term at the left-hand side of the super Filippov-Jacobi identity (\ref{LHS}) and the right-hand side can be found at the right-hand side of super Filippov-Jacobi identity (\ref{RHS}) (terms underlined by solid black line). Hence these terms can be cancelled. Analogously one can see that all underlined terms in (\ref{LHS}) are equal to the sum in the same way underlined terms in (\ref{RHS}).  

After all cancellations the left-hand side od the super Filippov-Jacobi identity has the form
\begin{equation}\label{left-remaining}
    -\mathunderline{blue}{ x[D(y),D(z)][u,v]}- \mathdashuline{blue}{ (-1)^{|x,yz|}y[D(z),D(x)][u,v]} - \mathdotuline{blue}{ (-1)^{|xy,z|}z[D(x),D(y)][u,v]}
\end{equation}
and the right-hand side of the same identity has the form
\begin{equation}\label{right-remaining}
    \begin{split}
        \mathdashuline{blue}{D(x)[y,D(z)[u,v]]} + \mathdotuline{blue}{(-1)^{|x,yz|}D(y)[z,D(x)[u,v]]}+ \mathunderline{blue}{(-1)^{|xy,z|}D(z)[x,D(y)[u,v]]}& \\ \mathdotuline{blue}{
         +(-1)^{|z,uv|}D(x)[D(y)[u,v],z]}+\mathunderline{blue}{ (-1)^{|x,yzuv|}D(y)[D(z)[u,v],x]}&\\ + \mathdashuline{blue}{ (-1)^{|xy,z|+|y,uv|}D(z)[D(x)[u,v],y]}.&
    \end{split}
\end{equation}
The terms that are underlined in the same way can be cancelled due to identity (\ref{identity 6}). This completes the proof of the super Filippov-Jacobi identity and the proof of the theorem.
\end{proof}

\section{CONCLUSIONS}
Given a commutative superalgebra and its even derivation, one can construct the transposed Poisson superalgebra using bracket (\ref{intr-graded-D-bracket}). In this paper we considered the question of how this approach can be extended to odd derivations of a commutative superalgebra. To solve this question, we need an analogue of bracket (\ref{intr-graded-D-bracket}) for the case of an odd derivation. This analogue is bracket (\ref{intr-graded-delta-bracket}). However this bracket is defined not on the elements of a commutative superalgebra but on the left supermodule over this algebra generated by an odd derivation, and we cannot omit an odd derivation (as in the case of an even derivation) without breaking the consistency of the parities. This leads us to consider a supermodule over a commutative superalgebra. We equip this supermodule with the bracket (\ref{intr-graded-delta-bracket}). This bracket induces a Jordan algebra structure on the subspace of even elements and a Jordan module structure on the subspace of odd elements. Thus, we come to the conclusion that if our goal is to treat odd derivations within a structure similar to a transposed Poisson superalgebra, we must, first, pass from a commutative superalgebra to a supermodule over this superalgebra and, second, use not a Lie superalgebra but a Jordan algebra and a Jordan module over this algebra. The latter structure is close in its properties to the notion of a Jordan superalgebra, and such a conjecture was made in \cite{Roger}. However, our calculations showed that the super Jordan identity does not hold in this case.

%% If you have bibdatabase file and want bibtex to generate the bibitems, please use
%\bibliographystyle{PEAS_numbers}  %% uncomment for the numerical unsorted style of the bibliography
%\bibliographystyle{PEAS_numbers_sort}  %% uncomment for the numerical alphabetically sorted style of the bibliography
%\bibliographystyle{PEAS_authoryear} %% uncomment for the author-year style of the bibliography
%\bibliography{<your bibdatabase(s)>}

%% else use the following coding to input the bibitems directly in the TeX file.

 %% comment if you use bibtex

%\end{multicols}  %% uncomment for two-column layout

%% Summary in Estonian (not obligatory for foreign authors)

\begin{summary}{Transponeeritud Poissoni superalgebra}{Viktor Abramov, Nikolai Sovetnikov}
\"Uhe olulise n\"aite transponeeritud Poissoni algebrast saab konstrueerida kommutatiivse algebra ja selle derivatsiooni abil. Seda l\"ahenemist saab laiendada superalgebratele, see t\"ahendab, et saab konstrueerida transponeeritud Poissoni superalgebra, kui on antud kommutatiivne superalgebra ja selle paarisderivatsioon. Selles artiklis n\"aitame, et paaritute derivatsioonide kaasamine selle l\"ahenemisviisi raamistikku n\~ouab uue struktuuri defineerimist. See on supervektorruum kahe tehtega, mis rahuldavad transponeeritud Poissoni superalgebra koosk\~ola tingimust. Esimese tehe m\"a\"arab vasakpoolne supermoodul \"ule kommutatiivse superalgebra ja teine on Jordani sulg. Seej\"arel t\~oestatakse, et kommutatiivse superalgebra paaritu derivatsiooni poolt tekitatud supervektorruumi korral k\~oik eespool mainitud struktuuri tingimused on t\"aidetud. On n\"aidatud, kuidas konstrueerida 3-Lie-superalgebrat, kui on antud transponeeritud Poissoni superalgebra ja selle paarisderivatsioon.
\end{summary}

\end{document}